%% file: modcom_isit_2023_arxiv.tex
\documentclass[conference, letterpaper]{IEEEtran}
\usepackage{amsthm, amssymb}
\newtheorem{definition}{Definition}[section]
\newtheorem{theorem}{Theorem}[section]
\newtheorem{lemma}{Lemma}[section]
\newtheorem{example}{Example}[section]

\usepackage[utf8]{inputenc} 
\usepackage[T1]{fontenc}
\usepackage{url}              
\usepackage{cite}             

\usepackage[cmex10]{amsmath}  
\usepackage{mleftright}       
\mleftright                   

\usepackage{graphicx}         
\usepackage{booktabs}         
\usepackage{glossaries}
\input{acronyms.tex}

\newcommand\blfootnote[1]{%
  \begingroup
  \renewcommand\thefootnote{}\footnote{#1}%
  \addtocounter{footnote}{-1}%
  \endgroup
}

\usepackage{xcolor}
\usepackage{tikz}
\begin{document}

\title{Semantic Communication of Learnable Concepts} 


%

\author{\IEEEauthorblockN{Francesco Pase$^{\star }$, Szymon Kobus$^{\dagger }$, Deniz G{\"u}nd{\"u}z$^{\dagger }$, Michele Zorzi$^{\star }$\medskip}
	\IEEEauthorblockA{
		$^{\star}$University of Padova, Italy. Email: \texttt{\{pasefrance, zorzi\}@dei.unipd.it}\\
			$^{\dagger}$Imperial College London, London, UK. Email: \texttt{\{szymon.kobus17, d.gunduz\}@imperial.ac.uk}}}

\maketitle
\blfootnote{This work  received  funding  from  the  UKRI (EP/X030806/1)  for the project  AIR (ERC-CoG). For the purpose of open access, the authors have applied a Creative Commons Attribution (CCBY) license to any Author Accepted Manuscript version arising from this submission.}

\begin{abstract}
We consider the problem of communicating a sequence of concepts, i.e., unknown and potentially stochastic maps, which can be observed only through examples, i.e., the mapping rules are unknown. The transmitter applies a learning algorithm to the available examples, and extracts knowledge from the data by optimizing a probability distribution over a set of models, i.e., known functions, which can better describe the observed data, and so potentially the underlying concepts. The transmitter then needs to communicate the learned models to a remote receiver through a rate-limited channel, to allow the receiver to decode the models that can describe the underlying sampled concepts as accurately as possible in their semantic space. After motivating our analysis, we propose the formal problem of communicating concepts, and provide its rate-distortion characterization, pointing out its connection with the concepts of empirical and strong coordination in a network. We also provide a bound for the distortion-rate function.
\end{abstract}

\begin{tikzpicture}[remember picture,overlay]
		\node[anchor=north,yshift=-10pt] at (current page.north) {\parbox{\dimexpr\textwidth-\fboxsep-\fboxrule\relax}{
				\centering\footnotesize This paper has been accepted for presentation at the 2023 IEEE International Symposium on Information Theory. \textcopyright 2022 IEEE.
				}};
\end{tikzpicture}

\section{Introduction and Motivation}
\label{sec:intro}

With the growing number of mobile devices and sensors, massive amounts of data are collected today at the edge of communication networks. On the one hand, this data is the fuel for training large learning models like \glspl{dnn};
on the other hand, these models need to be stored, compressed, and communicated over bandwidth limited channels to the cloud, and protected against security and privacy risks~\cite{xu_privacy:2014}. These issues are increasingly limiting the application of typical centralized training approaches. Various federated/distributed learning paradigms have emerged as potential solutions to mitigate these limitations, which allow the models to be locally trained, and then aggregated in a cloud or edge server without moving local private data~\cite{BrendanMcMahan2017}. The main paradigm shift in distributed learning is to move the models, rather than the data, throughout the network, providing better privacy guarantees and reducing the communication load. However, today even the sizes of such learned models are becoming a concern, as transmitting huge models back and forth for training or inference purposes can easily congest wireless networks, specifically when considering that the edge devices like mobile phones, cars, robots etc., are usually wirelessly connected to the network, and thus have limited bandwidth \cite{Jankowski:ISIT:22}.

Consequently, it is time to investigate, with the proper information-theoretic models and tools, the fundamental limits of communicating models over rate-limited channels, and not just raw data. To this end semantic communications, which concerns with the semantic aspect of the message, maps naturally to the transmission of these learning models~\cite{gunduz:sem:2023}. The communication fidelity of these models can be judged by how close the
behavior of the reconstructed model at the receiver is to the desired one,
rather than by the accuracy of the reconstruction in the parameter space~\cite{Jankowski:ISIT:22}.


\begin{figure}[t!]
    \def\vscale{3.6}
    \def\hscale{5}
    \def\mscale{0.9}
    
	\centering
    \begin{tikzpicture}

    \node (a) at (0,0) {$C_1,\dots,C_n\sim P_C$};
    \node (b) at (0,-\mscale) {$\big\{S_i=\{z_j\}^m_{j=1}\big\}$};
    \node (c) at (0,-\vscale+\mscale) {$\big\{P(h|S_i)=\mathcal{A}(S_i)\big\}^n_{i=1}$};
    \node (d) at (0,-\vscale) {$m=\{1,\dots,2^{nR}\}$};
    \node (e) at (\hscale,-\vscale) {$\{Q_i(h)\}^n_{i=1}$};
    \node (f) at (\hscale,-\mscale/2) {
        \begin{tabular}{c} 
        $p_h(Y|X) h\sim Q_1(h),X\sim C_1$ \\ $\dots$ \\ 
        $p_h(Y|X) h\sim Q_n(h),X\sim C_n$
        \end{tabular}
    };
    \node[font=\fontsize{11}{0}\selectfont] (alice) at (0,-\vscale-\mscale/2) {\textbf{Alice}};
    \node[left of=alice] {\includegraphics[width=13px]{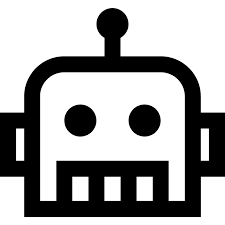}};
    \node[font=\fontsize{11}{0}\selectfont] (bob) at (\hscale,-\vscale-\mscale/2) {\textbf{Bob}};
    \node[right of=bob] {\includegraphics[width=13px]{img/robot.png}};


    \draw[->] (a) -- (b); 
    \draw[->] (b) -- (c);
    \draw[->] (c) -- (d);
    \draw[dashed,->] (d) -- node[above] {Noisy} (e);
    \draw[dashed,->] (d) -- node[below] {Channel} (e);
    \draw[->] (e) -- (f);
    \draw[ultra thick] (-1,-\vscale/2) -- (\hscale+1,-\vscale/2);
    \node[anchor=south, font=\fontsize{12}{0}\selectfont] at (\hscale/2,-\vscale/2+0.1) {Semantic Level};
    \node[anchor=north, font=\fontsize{12}{0}\selectfont] at (\hscale/2,-\vscale/2-0.1) {Technical Level};
    \end{tikzpicture}
    \vspace{-0.75cm}
    \caption{The problem of communicating concepts.}
    \vspace{-0.3cm}
	\label{fig:semscheme}
\end{figure}

\section{Related Work}
\label{sec:related}

The goal of this paper is to introduce the problem of communicating concepts, which is translated to that of conveying functions that better approximate them by learning their semantic aspects from data. The closest reference to this work is~\cite{havasi_2018}, in which the goal is to compress neural networks by applying bits back coding~\cite{frey:1997} to represent probability distributions over models with the minimum number of bits. The problem studied in~\cite{havasi_2018} is the single-shot version of our problem, focusing on the design of a practical coding scheme, which is called $\mathtt{MIRACLE}$, to efficiently compress neural networks. In \cite{Ofir:2016}, the authors study the connections between compressibility and learnability in the context of \gls{pac} learning, and show that the two concepts are equivalent when zero/one loss is considered, but not in the case of general loss functions. Another line of research investigates the connections between the generalization capabilities of learning algorithms, and the mutual information between the data and the model~\cite{xu:2017, Bassily:LittleIT:2018, Steinke:gen:2020}. The logic behind these results is to provide a bound on the generalization gap, i.e., the difference between the expected error and the training one, given some information-theoretic properties of the learning algorithm. However, if the environment imposes a constraint on such quantities, e.g., mutual information between the input and output of the learning rule, for example by introducing a rate-limited communication channel between the data and the final model, this influences not only the generalization gap, but also the training error itself (the output of the learning rule is constrained by the environment now), and so it is not clear how the gap between the best achievable test error changes as a function of the mutual information. This is the scenario studied in this work, where the constraint on the mutual information is not a property of the learning rule, but rather a physical limit imposed by the system. In~\cite{pase:jsait:2022}, a similar study is performed on the specific case of contextual multi-armed bandits, where the fundamental quantity is the mutual information $I(S; A)$ between the system states and the action taken by the agents, which is a property of the specific policy adopted. This work generalizes that idea to the supervised learning framework, and considers the effect of the communication rate $R$ on the final performance. It is also interesting to highlight the connections between this work and the study in~\cite{achille:ic:2021}, where the authors quantify the complexity of a learning algorithm output $Q$ with its Kullback–Leibler divergence from a prior model distribution $P$, which, in our system model, represents the minimum achievable rate to convey $Q$, when $P$ is set as the prior distribution. 

To conclude, this work is partially built on top of the results in~\cite{cuff:2010, CommDistribution, harsa:2010, functional_repr, pmlr-v162-theis22a}, which generalize the concept of rate-distortion theory~\cite{cover:IT} for standard data communication to probability distributions, where the fidelity requirement at the receiver is not to exactly reconstruct the input data, but rather to generate samples according to some input distribution. Indeed, here the semantic aspect of communication is captured by the fact that there is no need to convey the exact dataset sampled by the transmitter, but rather to represent with high fidelity the belief on the underlying concepts acquired after observing the data, which is the post-data probability distribution over the class of feasible models. 

\section{System Model}
\label{sec:system}

Let $\mathcal{E}$ denote the environment, i.e., the source, that generates a sequence of $n$ concepts, e.g., tasks, $\{ c_i \}_{i=1}^n$, $c_i \in \mathcal{C}$, sampled with probability $P_C$ in an \gls{iid} fashion. While $P_C$ is known by both Alice and Bob, neither of them can observe the sampled concepts directly. Alice has access to a sequence of $m$ samples $\{z_{i, j}\}_{j=1}^m$, where $z_{i, j} = (x_{i, j}, y_{i, j}) \in \mathcal{Z}$, sampled according to each of the concept distributions $p_{c_i}(Y|X) p_{c_i}(X), \; \forall i=1, \dots, n$. Alice and Bob agree on a hypothesis class, i.e., the model class $\mathcal{H}$, and on a pre-data coding probability distribution $P_h, \; \forall h \in \mathcal{H}$. We call the sequence $\{z_{i,j}\}_{j=1}^m$ of samples the dataset $s_i$. Alice applies a learning algorithm $\mathcal{A}: \mathcal{Z}^m \rightarrow \Phi(\mathcal{H})$ on $s_i$, which is a possibly stochastic function mapping a dataset to a probability distribution $Q_{h|s_i} = \mathcal{A}(s_i)$ over the set of models $\mathcal{H}$, and so, over the subset of all possible probability mappings $h: \mathcal{X} \rightarrow \Phi(\mathcal{Y})$, representing Alice's concept belief. With $\Phi(\mathcal{X})$ we denote the set of all possible probability distributions over the set $\mathcal{X}$. Consequently, the models are functions used by Alice and Bob to represent (or, more precisely, to approximate) the concepts' relation among data. To measure how well a model $h$ approximates a concept $c$, a per-sample loss $\ell_c(h, z) : \mathcal{H} \times \mathcal{Z} \rightarrow [0, 1]$ is defined, which compares the discrepancy between $p_c(y|x)$ and $h(x)$. 
In this work, we assume a bounded loss within $[0, 1]$ for the sake of clarity of exposition. Then, $\ell_c(Q, z) : \Phi(\mathcal{H}) \times \mathcal{Z} \rightarrow [0, 1]$ is the performance of the model belief $Q$, which is defined as $\ell_c(Q, z) = \mathbb{E}_{h \sim Q} \left[ \ell_c(h, z) \right]$. Upon observing the data, Alice can compute her empirical performance by using the \emph{empirical loss} on her dataset $S$ as 
\begin{align}
\label{eq:emp_loss}
    \tilde{\mathcal{L}}_C(Q, S) = \frac{1}{|S|} \sum_{z \in S} \ell_C(Q, z),
\end{align}
where $Q = Q_{H | S} = \mathcal{A}(S)$ is the post-data distribution inferred by Alice, given the data. We assume that, for any sequences of datasets $s^n,s^{\prime n}$, Alice's distribution can be factorized as $Q^n_{h^n|s^n}=\prod_{i=1}^n Q_{h_i|s_i}$ such that $s_i = s^\prime_j \Rightarrow Q_{h_i|s_i}=Q_{h_j|s^\prime_j}$. However, to assess how well the belief $Q$ represents the concept $C$, in machine learning we are usually interested in the \emph{true loss} 
\begin{align}
\label{eq:true_loss}
    \mathcal{L}_C(Q) = \mathbb{E}_{Z \sim C} \left[ \ell_C(Q, Z) \right],
\end{align} 
i.e., the expected performance on a new unseen sample.
Given the realization of the datasets $\{s_i\}_{i=1}^n$, the problem for Alice is then to convey a message to Bob through a constrained communication channel, which limits the maximum number of bits she can convey per model, so that Bob can use the received information to reconstruct models $\{ \hat{h}_i\}_{i=1}^n$ that can approximate the concepts $\{c_i\}_{i=1}^n$ by minimizing the loss on random samples $\{z_i\}_{i=1}^n$ distributed according to the sequence of concepts, i.e., the true loss in Equation~(\ref{eq:true_loss}). We can observe that the task for Bob is not to exactly reconstruct the sequence $\{ h_i\}_{i=1}^n$ sampled by Alice, but rather to obtain samples $\{ \hat{h}_i\}_{i=1}^n$ whose probability distributions are close to the target ones, i.e., $\{ Q_{h_i|s_i}\}_{i=1}^n$.

\textbf{Remark.} We now briefly discuss why we are interested in learning rules $\mathcal{A}(S)$ that output model distributions, rather than single-point solutions: 
\begin{itemize}
    \item  First of all, the case in which Alice finds a point-wise estimate of the best model $h^*$ is included as a special case $Q_{h|S} = \delta_{h^*}$.
    \item Alice may want to express her uncertainty around the best choice $h^*$, which may be intrinsic in the learning algorithm $\mathcal{A}$, through the distribution $Q_{h|S}$.
    \item Usually, optimization algorithms used to train \glspl{dnn}, like \gls{sgd}, are stochastic algorithms.
    \item When $\mathcal{H}$ is the set of all \glspl{dnn} $h_{\omega}$ with a specific architecture parameterized by the parameter vector $\omega$, there exist many vectors $\omega$ performing in the same way. Moreover, small perturbations to the parameters usually does not reduce the final performance. This means that it is not required for Bob to reconstruct the exact value of the network parameters, but rather a nearby or an equivalent solution, and this variability is represented by $Q_{h_{\omega}|S}$. More importantly, $Q_{h_{\omega}|S}$ can be exploited to reduce the rate needed to convey the models, thus saving network resources~\cite{havasi_2018}. This is the semantic aspect of communication captured by our framework, as the meaning of a concept $c$, i.e., the real unknown mapping, is conveyed through the model belief $Q_{h|S}$, whose loss expressed in Equation~(\ref{eq:true_loss}) quantifies its fidelity with respect to the real concept $c$. 
 \end{itemize}

\section{The Rate-Distortion Characterization }
 \label{sec:rate_distortion}

In this section, we first characterize the limit of the problem when $n=1$, i.e., one-shot concept communication, and then  generalize the problem to the $n$-sequence formulation. For the latter, two kinds of performance metrics are defined: the first one provides average performance guarantees, while the second one ensures the same performance guarantee for each sample $\hat{h}$. We will show that the minimum achievable communication rate that can guarantee a certain distortion level is the same in both cases, as long as sufficient common randomness between Alice and Bob is available. 

\subsection{Single-Shot Problem}
\label{sub:single_shot}

The single-shot version of the problem has been studied in \cite{havasi_2018}, where the authors propose $\mathtt{MIRACLE}$, a neural network compression framework based on bits back coding~\cite{frey:1997}, providing an efficient single-shot model compression scheme showing empirically that, with enough common randomness, it is possible to convey the model with an average of $K \simeq D_{\text{KL}}(Q || P)$ bits with very good performance, where $P$ is the pre-data coding model distribution, and $Q$ is the optimized post-data distribution, providing a belief over well-performing neural networks, or, equivalently, over a set of parameter vectors, as explained in Section~\ref{sec:system}. 
However, from Lemma~$1.5$ in~\cite{harsa:2010}, the average number of bits $\mathbb{E}_{C,S}\left[ K \right]$ needed to exactly code $Q$ with $P$, when sufficient common randomness is available, can be bounded by 
\begin{align}
    \mathcal{R}  \leq \mathbb{E}_{C,S}\left[ K \right] \leq \mathcal{R} + 
    &2 \log\left( \mathcal{R} + 1\right) + \mathcal{O}(1),
\end{align}
where $\mathcal{R}=\mathbb{E}_{C,S} \left[ D_{\text{KL}}(Q || P)\right]$,
while using exactly $\mathcal{R}$ bits may lead to samples which are distributed according to $\Tilde{Q}$, slightly different from the target $Q$~\cite{harsa:2010}.
More recent results \cite{functional_repr} allow to find even stronger guarantees (Corollary 3.4 in \cite{pmlr-v162-theis22a}) for this relationship: 
\begin{align}
    \mathbb{E}_{C,S}\left[ K \right] \leq  \mathcal{R} +
    \log \left( \mathcal{R} + 1 \right) + 4.
\end{align}
\subsection{$n$-Length Formulation}
\label{sub:n_formulation}

We now study the problem depicted in Figure~\ref{fig:semscheme}, when we let Alice code a sequence of $n$ concept realizations, i.e., datasets, and study the information-theoretic limit of the system as ${n\rightarrow \infty}$. Specifically, we are interested in the trade-off between the rate $R$, which is defined as the average number of bits consumed per model by Alice to convey the concept process to Bob, and the performance, which is the true loss that can be obtained by Bob (see Eq.~(\ref{eq:true_loss})). We start by defining the proper quantities involved.

\begin{definition}[Rate-Distortion Coding Scheme]
\label{def:coding_scheme}
A $(2^{nR}, n)$ coding scheme consists of an alphabet $\mathcal{X}$, a reconstruction alphabet $\mathcal{\hat{X}}$, an encoding function $f_n :\mathcal{X}^n \rightarrow \{ 1, 2, \ldots, 2^{nR} \}$, a decoding function $g_n: \{ 1, 2, \ldots, 2^{nR}\} \rightarrow \hat{\mathcal{X}}^n$, and a distortion measure $d: \mathcal{X}^n \times \mathcal{\hat{X}}^n \rightarrow \mathbb{R}^+$, comparing the fidelity between $x^n$ and $\hat{x}^n$. Specifically, we are interested in the expected distortion $\mathbb{E}\left[ d(x^n, \hat{x}^n)\right] = \sum_{x^n\in\mathcal{X}^n} p(x^n) d(x^n, g_n(f_n(x^n)))$.
\end{definition} 

\begin{definition}[Rate-Distortion] A rate-distortion pair $(R, \epsilon)$ is said to be achievable for a source $p(x)$ and a distortion measure $d$, if there exists a sequence of $(2^{nR}, n)$ rate-distortion coding schemes with
\begin{align}
    \lim_{n \rightarrow \infty} \mathbb{E}\left[ d(x^n, g_n(f_n(x^n)))\right] < \epsilon.
\end{align}
\end{definition}

However, as specified in the remark in Section~\ref{sec:system}, we are interested in conveying beliefs $Q \in \Phi(\mathcal{H})$, i.e., samples drawn according to the probability $Q$, obtained from the datasets $s^n \in \mathcal{S}^n$, where $\mathcal{S}^n = \mathcal{Z}^m$. In our case, the coding function $f_n: \mathcal{S}^n \rightarrow \{ 1, 2, \ldots, 2^{nR} \}$ maps the sequence $s^n = \{s_i\}_{i=1}^n$ to a message $f_n(s^n)$ from which Bob can obtain the models $\hat{h}^n = g_n(f_n(s^n))$. Our distortion then considers the difference between the $\hat{h}^n$'s distribution $\hat{Q}^n$, and the one achievable by Alice $Q^n = \{\mathcal{A}(s_i)\}_{i=1}^n$ by comparing their samples $\hat{h}^n$ and $h^n$. Consequently, we define with $\hat{Q}_{S^n, \hat{H}^n}$ (or simply $\hat{Q}^n$) the joint distribution between the datasets and models induced by a $(2^{nR}, n)$ coding scheme, whose marginals are $\hat{Q}_{S_i, \hat{H}_i}$ for $i = 1, \dots, n$.

\begin{definition}[Concept Distortion]
\label{def:single_concetp_dist} For the problem of communicating concepts, we define the following distortion on the model beliefs $Q$ and $\hat{Q}$:
\begin{align}
\label{eq:single_concept_dist}
    d_{sem}(Q, \hat{Q}) = \mathbb{E}_{C, S} \left[ \mathcal{L}_{C}(\hat{Q}) - \mathcal{L}_{C}(Q)  \right] .
\end{align}
\end{definition}

The rationale behind this definition is that $Q$, which is the target distribution at the transmitter, is optimized on a given dataset $S$ without any constraint. Therefore, it is reasonable to assume $\mathcal{L}_{C}(\hat{Q}) - \mathcal{L}_{C}(Q)$ to be always non-negative. We notice that $d_{sem}$ quantifies the gap between the concept reconstruction at the receiver and the one at the transmitter, which is a semantic measure on the unknown true loss $\mathcal{L}(Q)$.

\begin{definition}[$n$-Sequence Concept Distortion]
\label{def:concetp_dist} For the problem of communicating concepts, we define the following distortion on the sequence of the model beliefs $Q^n$ and $\hat{Q}^n$:
\begin{align}
\label{eq:concept_avg_dist}
    d_{avg}(Q^n, \hat{Q}^n) = \frac{1}{n} \sum_{i=1}^n d_{sem}(Q_i, \hat{Q}_i),
\end{align}
and
\begin{align}
\label{eq:concept_strong_dist}
    d_{max}(Q^n, \hat{Q}^n) =\max_{i=\{1, \ldots, n\}} d_{sem}(Q_i, \hat{Q}_i).
\end{align}
\end{definition}
In practice, $d_{avg}$ defines a constraint on the true loss achievable by Bob averaged over the performance of the sequence $\hat{h}^n \sim \hat{Q}^n$, whereas $d_{max}$ imposes a constraint on the loss of every marginal 
$\hat{h}_i \sim \hat{Q}^n_i, \; i \in \{1, \ldots, n\}$.

\textbf{Remark.} The distortion, easily defined for one-shot communications, can be generalized to a sequence of $n$ concepts in multiple ways.
Specifically, if one is interested in a system-level loss, then satisfying the constraint on $d_{avg}$ could be enough. However, to provide a per-model guarantee on the performance, then $d_{max}$ is the distortion to use.

\begin{definition}[Rate-Distortion Region]
The rate-distortion region for a source is the closure of the set of achievable rate-distortion pairs $(R, \epsilon)$.
\end{definition}

\begin{definition}[Rate-Distortion Function] The rate-distortion function $R(\epsilon)$ is the infimum of rates $R$ such that $(R, \epsilon)$ is in the rate-distortion region.
\end{definition}    

\subsection{Average Distortion $d_{avg}$}
\label{sub:avg_dist}

We first analyze the problem with the average distortion $d_{avg}$, as defined in Equation~(\ref{eq:concept_avg_dist}).

\begin{theorem}[Rate-Distortion Theorem for $d_{avg}$]
\label{cor:oracle_rate_distortion}
For the problem of communicating concepts with distortion $d_{avg}$, the rate-distortion function satisfies
\begin{align}
\label{eq:rate_dist_avg}
    R(\epsilon) = \min_{\substack{\tilde{Q}_{H|S}: \\ d_{sem}(Q, \tilde{Q}) \leq \epsilon }} I(S;H),
\end{align}
where $I(S;H)$ is the mutual information between the data $S$ and the model $H$~\cite{cover:IT}. 
\end{theorem}
\begin{proof}
    See Appendix~\ref{app:kramer_translate}
\end{proof}

\subsection{Maximum Distortion $d_{max}$}
\label{sub:max_dist}

In this case the distortion function implies a constraint on the performance of each symbol, i.e., model realization. First of all, we just provide a simple scheme in which
$\lim_{n \rightarrow \infty} d_{avg}(Q^n, \hat{Q}^n) = 0$ does not imply 
$\lim_{n \rightarrow \infty} d_{max}(Q^n, \hat{Q}^n) = 0$, meaning that in general a code that achieves $0$ distortion on average, may not achieve $0$ distortion model-wise, i.e., we cannot guarantee a single-model performance.





\begin{example}
\label{ex:weak_dist}
Let $\mathcal{H}=\{h_0, h_1\}$,
and performance
$\ell_c(h_0,z)=0, \ell_c(h_1,z)=1, ~ \forall z\in \mathcal{Z},$ $\forall\: c\in \mathcal{C}$.
Let $\forall\: c^n \in \mathcal{C}^n$, Alice's distribution $Q_i(h_j|S)=\frac{1}{2}$, where $j\in \{0,1\}$, and $Q^n=\prod_{i=1}^n {Q_i}$,
while Bob's distribution is deterministic $\hat{Q}_{2i}(h_0|S)=1, \hat{Q}_{2i+1}(h_1|S)=1$.
Then
\begin{align*}
     d_{avg}&(Q^n, \hat{Q}^n) = \\ 
     &\begin{cases}
    \frac{1}{n} \sum_{i=1}^\frac{n}{2} \left( 1-\frac{1}{2}+0-\frac{1}{2} \right) = 0,& \text{if } n \text{ is even} \\
    \frac{1}{n}[\sum_{i=1}^\frac{n-1}{2} \left( 1-\frac{1}{2} +0-\frac{1}{2} \right) + 1-\frac{1}{2}]=\frac{1}{2n},              & \text{otherwise}
\end{cases}&
\end{align*}
Thus, {\centering $\lim_{n \rightarrow \infty} d_{avg}(Q^n, \hat{Q}^n) = 0$}

\noindent but
\begin{align*}
      d_{max}(Q^n, \hat{Q}^n) = \max \left\{1-\frac{1}{2}, 0-\frac{1}{2}\right\} = \frac{1}{2} \\
      \implies \lim_{n \rightarrow \infty} d_{max}(Q^n, \hat{Q}^n) \neq 0.
\end{align*}

\end{example}

The question now is what is needed to ensure the same distortion $\epsilon$ to the single-model performance, i.e., to ensure $d_{max} < \epsilon$, which is of particular interest in the semantic communication of concepts like the one considered here. 
We now distinguish between the two ways in which $\hat{Q}^n$ can converge to a target $Q^n$ -- empirical and strong.
These two notions, introduced later, map precisely onto the difference between convergence of $d_{avg}$ and $d_{max}$.
The focus is changed from determining which distortion is achievable to what joint distributions of $H$ and $S$ are feasible under some rate constraint.


First, we extend the encoding and decoding functions $f_n$ and $g_n$ in Definition~\ref{def:coding_scheme} to accept an additional common input $\omega \in \Omega$, which is generated by a source of common randomness $p(\omega)$. We define a $(2^{nR}, 2^{nR_0}, n)$ stochastic code consisting of functions $f_n :\mathcal{S}^n \times \{1, \ldots, 2^{nR_0}\} \rightarrow \{ 1, 2, \ldots, 2^{nR} \}$ and $g_n : \{1, 2, \ldots, 2^{nR} \} \times \{1, \ldots, 2^{nR_0}\} \rightarrow \hat{\mathcal{H}}^n$, which consumes on average $R_0$ bits of common randomness per sample. 

\begin{definition}
    A desired distribution $Q_{S, H}$ is achievable for empirical coordination with rate pair $(R, R_0)$ if there exists a sequence of $(2^{nR}, 2^{nR_0}, n)$ codes and a choice of common randomness distribution $p(\omega)$ such that
    \begin{align}
        \text{TV}\left(\hat{Q}_{s^n\hat{h}^n}, Q_{S, H} \right) \rightarrow 0,
    \end{align}
where $\text{TV}(\hat{Q}, Q )$ indicates the total variation between distributions $\hat{Q}$ and $Q$, $\hat{Q}_{s^n\hat{h}^n}(s, \hat{h}) = \frac{1}{n} \sum_{i=1}^n \mathbf{1}_{(s_i, \hat{h}_i) = (s, \hat{h})}$, and $\mathbf{1}_{\mathcal{I}} = 1$ if the condition $\mathcal{I}$ is true, and $0$ otherwise.
\end{definition}
In other words, the empirical coordination property requires that the joint empirical distribution of the pairs $(s_i, g_n(f_n(s^n))_i)$ converges, in total variation, to the desired distribution. Notice how, in Example \ref{ex:weak_dist}, the joint distributions $Q^n$ and $\hat{Q}^n$ converge in their empirical distributions, while differing letter-wise.

As already pointed out in~\cite{cuff:2010}, introducing common randomness does not improve the performance of empirical coordination schemes, meaning that any distribution achievable for empirical coordination by a $(2^{nR}, 2^{nR_0}, n)$ coding scheme is also achievable with $R_0 = 0$. Moreover, empirical coordination schemes can be used to construct rate-distortion schemes for $d_{avg}$~\cite{cuff:2010}. However, as observed in Example~\ref{ex:weak_dist}, they do not equate to the same per-symbol performance requirements.

\begin{definition}
    A desired distribution $Q_{D, H}$ is achievable for strong coordination with rate pair $(R, R_0)$ if there exists a sequence of $(2^{nR}, 2^{nR_0}, n)$ coordination codes and a choice of common randomness distribution $p(\omega)$ such that
    \begin{align}
        \text{TV}\left(\hat{Q}_{s^n\hat{h}^n}, \prod_{i=1}^n Q_{s_i, h_i} \right) \rightarrow 0,
    \end{align}
where $\hat{Q}_{s^n\hat{h}^n}$ is the joint distribution induced by the stochastic coding scheme.
\end{definition}

\begin{lemma}[$d_{max}$ Achievability]
\label{lem:achievable_max}
    With sufficient common randomness, the rate-distortion region $(R, \epsilon)$ for $d_{max}$ is the same as the one for $d_{avg}$.
\end{lemma}
\begin{proof}
    See Appendix~\ref{app:achievable_max}.
\end{proof}
Given a constraint on the distortion $\epsilon$ for $d_{avg}$ achievable with minimum rate of $R_{\epsilon}$, if Alice and Bob can use common randomness, then it is possible to satisfy, at the same rate $R_{\epsilon}$, the same level of distortion $\epsilon$ for $d_{max}$. To translate it into machine learning parlance, we showed that the communication rate needed to provide some performance guarantees on the expected average system test error, 
and on the expected single-model performance, is the same, as long as sufficient common randomness is available. 
In both cases, the characterization is over the expected performance, thus for any one realization the $k$-th model might have higher than desired loss.

\subsection{Coding Without the Marginal $Q_H$}
\label{sub:kl_achievable}

We notice that all the previous achievability results assume knowledge of the exact marginal ${Q_H = \sum_{c \in \mathcal{C}, s \in \mathcal{Z}^m} Q_{H|s}P_{s|c}P_c}$ to be used as pre-data coding model distribution (see Section~\ref{sec:system}), which is usually not known and difficult to obtain. Consequently, we are interested in studying the minimum achievable rate, when a generic coding distribution $P_H$ is used to code $Q_{H|S}$.

\begin{theorem}[Achievability with General $P_H$]
\label{thm:rate_dist_kl}
For the problem of communicating concepts, the minimum achievable rate for both $d_{avg}$ and $d_{max}$ with pre-data coding distribution $P_H$ is
\begin{align}
\label{eq:rate_dist_kl}
    R(\epsilon) = \min_{\substack{\tilde{Q}_{H|S}: \\  d_{sem}(Q, \tilde{Q}) \leq \epsilon }} \mathbb{E}_{C, S} \left[ D_{\text{KL}} \left( \tilde{Q}_{H|S} \| P_H \right) \right],
\end{align}
assuming sufficient common randomness. 
\end{theorem}
\begin{proof}
See Appendix~\ref{app:rate_dist_kl}.
\end{proof}

It is known that when $P_H = Q_H$, i.e., using the marginal, Equation~(\ref{eq:rate_dist_kl}) is minimized, and so when the marginal is not known we pay an additional penalty given by $\mathbb{E}_{C, S} \left[ D_{\text{KL}} \left( Q_{H|S} \| P_H \right) \right] - I(S; H) = \mathbb{E}_{C, S} \left[ D_{\text{KL}} \left( Q_{H} \| P_H \right) \right]$.

\section{Communicating the Data vs Communicating the Model}
 \label{sec:model_data_comm}

To motivate our research problem, we comment on the advantages for Alice of first compressing a trained model, and then sending it to Bob (\emph{scheme 1}), versus a second framework (\emph{scheme 2}), in which Alice communicates a compressed version of the dataset $\hat{S}^2 = \rho(S)$, using which Bob trains his models. Given a quantity $X$, we indicate with $X^i$ the same quantity in the $i$-th scheme. First of all, we see that in scheme 1, the corresponding Markov chain is $C \rightarrow S \xrightarrow{R} \hat{H}^1$, where the $R$ above the arrow indicates the information bottleneck between the two random variables. However, the same chain for the second scheme reads $C \rightarrow S \xrightarrow{R} \hat{S}^2 \rightarrow \hat{H}^2$. By the data processing inequality~\cite{cover:IT}, the rate constraint imposes $I(S; \hat{H}) \leq R$ in both cases, limiting the set of all feasible beliefs $Q \in \Phi(\hat{\mathcal{H}})$. 

Now, we reasonably assume that the optimal solution $Q^*(S, H)$ constrained to $I(S; H) \leq R$ lies on the boundary of the constraint, i.e., $I_{Q^*}(S; H) = R$, where $I_{Q^*}(S; H)$ indicates that the mutual information is computed using $Q^*$. In this case we can see that for scheme 1, $I(S; \hat{H}^1) = R$, whereas for scheme 2, $I(S; (\hat{S}^2, \hat{H}^2)) = R$. Indeed, in the former scheme Alice conveys 
 just the random variable $\hat{H}^1$, i.e., the model, to Bob, whereas in the latter, the pair $(\hat{S}^2, \hat{H}^2)$ is communicated. However, the information bottleneck is the same. Consequently, we obtain
\begin{align}
    I(S; \hat{H}^1) = I(S;\hat{H}^2) + I(S; \hat{S}^2 | \hat{H}^2).
\end{align}
By non-negativity of the mutual information, we always have $I(S; \hat{H}^1) \geq I(S;\hat{H}^2)$, meaning that the rate constraint on the communication channel translates differently into a model constraint for the two schemes, being stricter for the second one. In particular, the gap between the two schemes is exactly $I(S; \hat{S}^2 | \hat{H}^2)$. Consequently, for scheme 2 the optimal compression function $\rho(S)$ must achieve $I(S; \hat{S}^2 | \hat{H}^2) = 0$, and so the only way to match the performance of scheme 1 is to account for the optimal distribution $Q_{\hat{H}|S}$ when computing $\hat{S}$. 
\subsection{Distortion Rate Bound}
\label{sub:distortion_rate_bound}

In this section, we bound the distortion-rate function for $d_{max}$, which is useful to translate the rate constraint into a performance gap.
We now define $\Delta_R = R^* - R$ to be the difference between the rate $R^*$ of the optimal distribution and the rate $R$ of the channel imposed by the problem. In general, we assume $R^* \geq R$ and $\mathbb{E}_{C, S} \left[\mathcal{L}(Q^*) \right] \leq \mathbb{E}_{C, S} \left[\mathcal{L}(Q) \right]$, where $Q^*$ is achievable with rate $R^*$, and $Q$ with rate $R$.

\begin{lemma}
\label{lem:dist_rate_bound}
Assuming that $\ell_C(z, h)$ is upper bounded by $L_{\text{max}} \; \forall h \in \mathcal{H}, \; \forall z \in \mathcal{Z}, \; \forall C \in \mathcal{C}$, and that distortion $d_{max}$ is considered, the distortion-rate function for the problem of communicating concepts when using scheme 1 can be upper bounded by
\begin{align}
\label{eq:dist_rate_bound}
    \epsilon^1(\Delta_R) \leq L_{\text{max}} \cdot \min \Big\{\sqrt{\frac{1}{2} \Delta_R},  \sqrt{1 - e^{-\Delta_R}} \Big\},
\end{align}
and when using scheme 2 by
\begin{align}
\label{eq:dist_rate_bound_2}
\begin{split}
    \epsilon^2(\Delta_R) \leq L_{\text{max}} \cdot \min \Big\{ & \sqrt{\frac{1}{2} \left( \Delta_R +I(S;\hat{S}^2|\hat{H}^2)\right)}, \\
    & \sqrt{1 - e^{- \left( \Delta_R +I(S;\hat{S}^2|\hat{H}^2) \right)}} \Big\}.
\end{split}
\end{align}
\end{lemma}
\begin{proof}
    See Appendix~\ref{app:dist_rate_bound}.
\end{proof}

\section{Conclusion}
\label{sec:conclusion}

We introduced the problem of conveying concepts, where concepts naturally appear as sequences of samples and can be approximated by learnable models, as in standard statistical learning. We study the framework by applying information-theoretic tools to the problem of communicating many models jointly. We characterized its rate-distortion function for two different notions of system-level distortion, provided a bound for the distortion-rate function, and argued why jointly learning, compressing, and communicating models should be preferred over compressing and sending the datasets. For future investigations, the plan is to study in detail the relations between model compression and test accuracy, and to design practical model communications schemes. 


\newpage

\bibliographystyle{IEEEtran}
\bibliography{biblio}


\newpage
\hphantom{1em}
\newpage 

\appendices
\section{Proof of Theorem~\ref{cor:oracle_rate_distortion}}
\label{app:kramer_translate}

To prove  Theorem~\ref{cor:oracle_rate_distortion}, we need to show that we can translate the average distortion $d_{avg}$ requirement into a constraint on the empirical distribution $\hat{Q}_{s^n\hat{h}^n}(d, h) = \frac{1}{n} \sum_{i=1}^n \mathbf{1}_{(s_i, h_i) = (s, h)}$, where $\mathbf{1}_{\mathcal{I}} = 1$ if the condition $\mathcal{I}$ is true, and $0$ otherwise. We can write 
\begin{align*}
    d_{avg} & (Q^n, \hat{Q}^n) = \frac{1}{n} \sum_{i=1}^n d(Q_i, \hat{Q}_i)\\
    & = \mathbb{E}_{C, S} \left[ \frac{1}{n} \sum_{i=1}^n \mathcal{L}_{C}(\hat{Q}_i) - \mathcal{L}_{C}(Q_i)  \right] \\
    & \stackrel{(a)}{=} \mathbb{E}_{C, S} \left[ \frac{1}{n} \sum_{i=1}^n \sum_{h \in \mathcal{H}} \mathcal{L}_{C}(h) (\mathbf{1}_{(S, \hat{h}_i) = (S, h)} - Q_{h|S}) \right] \\
    & \stackrel{(b)}{=} \mathbb{E}_{C, S} \left[  \sum_{h \in \mathcal{H}} \mathcal{L}_{C}(h) \left(\frac{1}{n} \sum_{i=1}^{n}  \mathbf{1}_{(S, \hat{h}_i) = (S, h)} - Q_{h|S}\right)  \right] \\
    & = \mathbb{E}_{C, S} \left[  \sum_{h \in \mathcal{H}} \mathcal{L}_{C}(h) \left(\hat{Q}_{s^{n'}\hat{h}^{n'}}(d, h) - Q_{h|S}\right)  \right] \\
        & = \mathbb{E}_{C, S} \left[  \mathcal{L}(\hat{Q}_{s^{n'}h^{n'}}) - \mathcal{L}(Q)  \right] \\
    & = d_{sem}(Q, \hat{Q}_{s^{n'}\hat{h}^{n'}}),
\end{align*}
where the indicator function in (a) depends on the chosen $(2^{nR}, n)$ coding scheme, and (b) comes from the fact that Alice's sampling scheme does not depend on the model index $i$. We then showed that the average distortion requirement translates into a distortion between the empirical distribution $\hat{Q}_{s^{n}\hat{h}^{n}}$, and the target $Q$. Given this, Theorem 1 in~\cite{CommDistribution} applies, and our Theorem~\ref{cor:oracle_rate_distortion} follows.

\section{Proof of Lemma~\ref{lem:achievable_max}}
\label{app:achievable_max}
\begin{proof}
    Theorem~\ref{cor:oracle_rate_distortion} provides the minimum achievable rate to satisfy $d_{avg} \leq \epsilon$, which is provided by an optimized distribution $Q_{S,H}^*$  that minimizes the mutual information $I(S;H)$, and satisfies $d_{avg} \leq \epsilon$. Now, we impose strong coordination between Alice and Bob, by using as target joint distribution $Q_{S^n,H^n}^* = \prod_{i=1}^n Q_{S,H}^*$, and by Theorem 10 of~\cite{cuff:2010} the same rate $I(S;H)$ of empirical coordination can be achieved, as long as enough common randomness is available. Now, given that Bob's distribution $\hat{Q}_{S^n, \hat{H}^n}$ converges in total variation to $Q_{S^n,H^n}^*$, so it happens for the marginal, i.e., $\hat{Q}_{S^n, \hat{H}^n}(s_i, \hat{h}_i) \xrightarrow{\text{TV}} Q_{S,H}^*$. However, by construction, $Q_{S,H}^*$ satisfies the distortion constraint symbol-wise, and so satisfies $d_{max}$.
\end{proof}

\section{Proof of Theorem~\ref{thm:rate_dist_kl}}
\label{app:rate_dist_kl}
 To prove Theorem~\ref{thm:rate_dist_kl}, we observe that Alice now needs to convey her probability distribution $Q^n_{H^n | S^n}$ to Bob using $P^n_H = (P_{H^n})^n$ to code it. If we indicate with $\mathbb{E}_{C,S}\left[ K \right]$ the average number of bits spent to convey model belief $Q$ using distribution $P$, Corollary 3.4 in~\cite{pmlr-v162-theis22a} provides the single-shot bounds
\begin{align*}
    \mathcal{R}(Q, P) \leq \mathbb{E}_{C,S}\left[ K \right] \leq  \mathcal{R}(Q, P) +
    \log \left( \mathcal{R}(Q, P) + 1 \right) + 4,
\end{align*}
where $\mathcal{R}(Q, P) = \mathbb{E}_{C,S} \left[ D_{\text{KL}}(Q || P)\right]$. 

We now translate the results to the $n$-length sequence, and compute its limit as $n$ grows indefinitely.
\begin{align*}
    \mathcal{R}(Q^n, P^n) & = \mathbb{E}_{C^n,S^n} \left[ D_{\text{KL}}(Q^n || P^n)\right] \\
    & = \mathbb{E}_{C^n} \left[ \mathbb{E}_{S^n|C^n} \left[ \mathbb{E}_{H^n|S^n} \log \frac{Q^n_{H^n|S^n}}{P^n_{H^n}}\right]\right] \\
    & \stackrel{(a)}{=}  \mathbb{E}_{C^n} \left[ \mathbb{E}_{S^n|C^n} \left[ \mathbb{E}_{H^n|S^n} \log \frac{\prod_{i=1}^n Q_{H|S_i}}{(P_{H})^n}\right]\right] \\
    & \stackrel{(b)}{=}  \mathbb{E}_{C^n} \left[ \mathbb{E}_{S^n|C^n} \left[ \mathbb{E}_{H^n|S^n} \log \frac{(Q_{H|S})^n}{(P_{H})^n}\right]\right] \\
    & = n \mathcal{R}(Q, P),
\end{align*}
where (a) is because samples are independent given the nature of the problem, and (b) is because, given the dataset realization, they are also identically distributed by our assumption. Consequently, we can upper bound the total number of bits with
\begin{align*}
 \mathbb{E}_{C,S}\left[ K \right] \leq n \mathcal{R}(Q, P) +
    \log \left( n \mathcal{R}(Q, P) + 1 \right) + 4,
\end{align*}
obtaining an average rate of
\begin{align*}
    \lim_{n \rightarrow \infty} \frac{\mathbb{E}_{C,S}\left[ K \right]}{n} \rightarrow  \mathcal{R}(Q, P).    
\end{align*}

\section{Proof of Lemma~\ref{lem:dist_rate_bound}}
\label{app:dist_rate_bound}

We now prove the result in Lemma~\ref{lem:dist_rate_bound}. Let $Q$ be the distribution at the sender, and $\hat{Q}$ the one at the receiver, i.e., the distribution $\hat{Q} = Q^*_{S, H}$ in the proof of Lemma~\ref{lem:achievable_max}. We first bound, for each $C \in \mathcal{C}$, $d_{max}$ by noticing that, for each $i \in \{1, \ldots, n\}$, the following holds
\begin{equation*}
\begin{split}
    & d_{sem} (Q^n_i, \hat{Q}^n_i) = \\
    & \stackrel{(a)}{=} \mathbb{E}_{C, S} \left[ \mathcal{L}_{C}(\hat{Q}) - \mathcal{L}_{C}(Q)  \right]\\
    & = \mathbb{E}_{C,S} \left[ \sum_{h \in \mathcal{H}} \mathbb{E}_{z \sim C} \left[ \ell(h, z)\right] \left(\hat{Q}_{\hat{h}|S} - Q_{h|S}\right)  \right] \\
    & \leq L_{\text{max}} \cdot \mathbb{E}_{C,S} \left[ \sum_{h \in \mathcal{H}} \hat{Q}_{\hat{h}|S} - Q_{h|S}\right]  \\
    & \leq L_{\text{max}} \cdot \mathbb{E}_{C,S} \left[ \text{TV}(\hat{Q}, Q) \right] \\
    &\stackrel{(b)}{\leq} L_{\text{max}} \cdot \mathbb{E}_{C,S} \Big[\min \Big\{  \sqrt{\frac{1}{2} 
    D_{\text{KL}}(\hat{Q} \| Q)}, 
     \sqrt{1 - e^{- D_{\text{KL}}(\hat{Q} \| Q)}} \Big\} \Big] \\
    & \stackrel{(c)}{\leq} L_{\text{max}} \cdot \min \Big\{ \sqrt{\frac{1}{2} \mathbb{E}_{C,S} \left[ D_{\text{KL}}(\hat{Q} \| Q)\right]}, 
    \sqrt{1 - e^{- \mathbb{E}_{C,S} \left[ D_{\text{KL}}(\hat{Q} \| Q)\right]}} \Big\}
\end{split}
\end{equation*}
where (a) is given by strong coordination, (b) is by the combination of the inequalities due to Pinsker and Breatgnolle-Huber, and (c) is Jensen's inequality. We now proceed to bound $ \mathbb{E}_{C,S} \left[ D_{\text{KL}}(Q \| \hat{Q})\right]$ by observing that the average distortion between $Q$ and $\hat{Q}$ is a linear function of the probabilities $\hat{Q}$, and the set $\hat{\mathcal{Q}} = \{ \hat{Q} \in \Phi(\mathcal{H}) : d_{sem}(Q, \hat{Q}) \leq \epsilon\}$ satisfying the constraint is convex, and that by the definition of Distortion-Rate function, $\hat{Q}$ is the distribution in $\hat{\mathcal{Q}}$ minimizing $\mathbb{E}_{C, S} \left[ D_{\text{KL}}(\hat{Q} \| P) \right]$, where $P$ is the pre-data coding distribution (see Theorem~\ref{thm:rate_dist_kl}). Then, Theorem 11.6.1 in ~\cite{cover:IT}, also known as the Pythagorean Theorem for the Kullback-Leibler divergence, applies, and we can bound
\begin{align*}
    \mathbb{E}_{C, S} \left[ D_{\text{KL}}(Q \| \hat{Q}) \right] & \leq \mathbb{E}_{C, S} \left[ D_{\text{KL}}(Q \| P) \right] - \mathbb{E}_{C, S} \left[ D_{\text{KL}}(\hat{Q} \| P) \right]  \\
    & \stackrel{(a)}{=} R^* - R \\
    & = \Delta_R,
\end{align*}
where (a) is given by Theorem~\ref{thm:rate_dist_kl}. Combined together, the two inequalities provide Equation~(\ref{eq:dist_rate_bound}). Regarding Scheme 2, it is sufficient to notice that $R = I(S;\hat{H}^2) + I(S; \hat{S}^2 | \hat{H}^2)$ (see Section~\ref{sec:model_data_comm}), from which Equation~(\ref{eq:dist_rate_bound_2}) follows.

We notice that this does not directly apply to $d_{avg}$, as the distribution $Q^*$ solving Equation~(\ref{eq:rate_dist_avg}) and providing the rate is not, in general, the one used by Bob to sample actions. This is true for the scheme for $d_{max}$, in which common randomness is introduced.

\end{document}

%% file: acronyms.tex
\newacronym{iot}{IoT}{Internet of Things}
\newacronym{ml}{ML}{machine learning}
\newacronym{dl}{DL}{Deep Learning}
\newacronym{marl}{MARL}{multi-agent reinforcement learning}
\newacronym{rl}{RL}{reinforcement learning}
\newacronym{decpomdp}{Dec-POMDP}{Decentralized Partially Observable Markov Decision Process }
\newacronym{pomdp}{POMDP}{Partially Observable Markov Decision Process}
\newacronym{uav}{UAV}{Unmanned Aerial Vehicle}
\newacronym{dqn}{DQN}{Deep Q-Network}
\newacronym{dnn}{DNN}{deep neural network}
\newacronym{dial}{DIAL}{Differentiable Inter-Agent Learning}
\newacronym{mdp}{MDP}{Markov decision process}
\newacronym{fov}{FoV}{Field of View}
\newacronym{cnn}{CNN}{Convolutional Neural Network}
\newacronym{nn}{NN}{neural network}
\newacronym{ddql}{DDQL}{Distributed Deep Q-Learning}
\newacronym{pdf}{PDF}{Probability Density Function}
\newacronym{ndpomdp}{ND-POMDP}{Networked Distributed Partially Observable Markov Decision Process}
\newacronym{radam}{RAdam}{Rectified Adam}
\newacronym{cdf}{CDF}{cumulative distribution function}
\newacronym{mpc}{MPC}{Model Predictive Control}
\newacronym{rv}{rv}{Random Variable}
\newacronym{qoe}{QoE}{Quality of Experience}
\newacronym{tlc}{TLC}{Telecommunications}
\newacronym{cml}{CML}{communications for machine learning}
\newacronym{mlc}{MLC}{machine learning for communications}
\newacronym{drl}{DRL}{Deep Reinforcement Learning}
\newacronym{rf}{RF}{Radio Frequency}
\newacronym{urllc}{URLLC}{Ultra-Reliable and Low-Latency Communications}
\newacronym{fl}{FL}{federated learning}
\newacronym{kpi}{KPI}{Key Performance Indicators}
\newacronym{mec}{MEC}{Mobile Edge Computing}
\newacronym{ei}{EI}{Edge Intelligence}
\newacronym{bs}{BS}{base station}
\newacronym{sdn}{SDN}{Software Defined Networking}
\newacronym{mimo}{MIMO}{Multiple-Input Multiple-Output}
\newacronym{gp}{GP}{Gaussian Process}
\newacronym{iiot}{IIoT}{Industrial Internet of Things}
\newacronym{csi}{CSI}{Channel State Information}
\newacronym{sgd}{SGD}{stochastic gradient descent}
\newacronym{iid}{i.i.d.}{independent and identically distributed}
\newacronym{ofdm}{OFDM}{Orthogonal Frequency Division Multiplexing}
\newacronym{los}{LOS}{Line-of-Sight}
\newacronym{nlos}{NLOS}{Non-Line-of-Sight}
\newacronym{snr}{SNR}{Signal to Noise Ratio}
\newacronym{rb}{RB}{Resource Block}
\newacronym{6g}{6G}{sixth generation}
\newacronym{ai}{AI}{artificial intelligence}
\newacronym{sfl}{SFL}{Synchronous Federated Learning}
\newacronym{frfl}{FRFL}{Fixed Rate Federated Learning}
\newacronym{pgm}{PGM}{Probabilistic Graphical Model}
\newacronym{hmm}{HMM}{Hidden Markov Model}
\newacronym{elbo}{ELBO}{Evidence Lower Bound}
\newacronym{pmf}{PMF}{Probability Mass Function}
\newacronym{smab}{SMAB}{Stochastic Multi-Armed Bandit}
\newacronym{mab}{MAB}{multi-armed bandit}
\newacronym{mc}{MC}{Monte Carlo}
\newacronym{is}{IS}{Importance Sampling}
\newacronym{dms}{DMS}{discrete memoryless source}
\newacronym{ucb}{UCB}{upper confidence bound}
\newacronym{ser}{SER}{Symbol Error Rate}
\newacronym{sc}{SC}{Semantic Communications}
\newacronym{voi}{VoI}{Value of Information}
\newacronym{nlp}{NLP}{natural language processing}
\newacronym{ts}{TS}{Thompson Sampling}
\newacronym{cmab}{CMAB}{contextual multi-armed bandit}
\newacronym{rccmab}{RC-CMAB}{rate-constrained \gls{cmab}}
\newacronym{rcmab}{R-CMAB}{remote \gls{cmab}}
\newacronym{ib}{IB}{information bottleneck}
\newacronym{merl}{MERL}{maximum entropy reinforcement learning}
\newacronym{pac}{PAC}{probably approximately correct}
\newacronym{aoi}{AoI}{Age of Information}
\newacronym{aoii}{AoII}{Age of Incorrect Information}
\newacronym{uoi}{UoI}{Urgency of Information}
\newacronym{vqvae}{VQ-VAE}{Vector Quantized Variational Autoencoder}
\newacronym{vae}{VAE}{Variational Autoencoder}
\newacronym{psnr}{PSNR}{Peak Signal to Noise Ratio}
\newacronym{mse}{MSE}{Mean Square Error}
\newacronym{lstm}{LSTM}{Long Short-Term Memory}